\documentclass[conference]{IEEEtran}
\IEEEoverridecommandlockouts
\usepackage{cite}
\usepackage{amsmath,amssymb,amsfonts, amsthm}
\usepackage{algorithm}
\usepackage[noend]{algorithmic}
\usepackage{graphicx}
\usepackage{subfigure}
\usepackage{float}
\usepackage{subfloat}
\usepackage{textcomp}
\usepackage{xcolor}
\usepackage{float}
\usepackage{footnote}
\def\BibTeX{{\rm B\kern-.05em{\sc i\kern-.025em b}\kern-.08em
    T\kern-.1667em\lower.7ex\hbox{E}\kern-.125emX}}

\newtheorem{asu}{Assumption}
\newtheorem{defn}{Definition}
\newtheorem{lem}{Lemma}
\newtheorem{rem}{Remark}

\newcommand\blfootnote[1]{%
  \begingroup
  \renewcommand\thefootnote{}\footnote{#1}%
  \addtocounter{footnote}{-1}%
  \endgroup
}

\begin{document}

\title{Connectivity-Aware UAV Path Planning with Aerial Coverage Maps}

\author{\IEEEauthorblockN{Hongyu Yang$^ \ast$, Jun Zhang$^ \dag$, S.H. Song$^ \ast$, and Khaled B. Lataief$^ \ast$, \textit{Fellow, IEEE} }
\IEEEauthorblockA{$^ \ast$Department of ECE, Hong Kong University of Science and Technology, Hong Kong \\ $^ \dag$Department of EIE, The Hong Kong Polytechnic University, Hong Kong}
Email: $^ \ast$\{hyangbb, eejzhang, eeshsong, eekhaled\}@ust.hk, $^ \dag$jun-eie.zhang@polyu.edu.hk
}

\maketitle

\begin{abstract}

    Cellular networks are promising to support effective wireless communications for unmanned aerial vehicles (UAVs), which will help to enable various long-range UAV applications. However, these networks are optimized for terrestrial users, and thus do not guarantee seamless aerial coverage. In this paper, we propose to overcome this difficulty by exploiting controllable mobility of UAVs, and investigate connectivity-aware UAV path planning. To explicitly impose communication requirements on UAV path planning, we introduce two new metrics to quantify the cellular connectivity quality of a UAV path. Moreover, aerial coverage maps are used to provide accurate locations of scattered coverage holes in the complicated propagation environment. We formulate the UAV path planning problem as finding the shortest path subject to connectivity constraints. Based on graph search methods, a novel connectivity-aware path planning algorithm with low complexity is proposed. The effectiveness and superiority of our proposed algorithm are demonstrated using the aerial coverage map of an urban section in Virginia, which is built by ray tracing. Simulation results also illustrate a tradeoff between the path length and connectivity quality of UAVs.

\end{abstract}

\section{Introduction}
    \blfootnote{This work was supported by the Hong Kong Research Grants Council under Grant No. 16209418.}

    \vspace{-0.3cm}
     \looseness=-1
     Unmanned aerial vehicles (UAVs) are becoming increasingly important in various civilian applications \cite{b1}. For their effective operation, wireless communications between UAVs and ground control stations (GCSs) is essential to transmit information such as flight status, control commands and sensing messages. As cellular networks have the advantages of providing wide-area, high-throughput, reliable and secure communications, connecting UAVs with cellular technology has received significant attention from both academia and industry \cite {b3,b2,b5}. With effective cellular connectivity, we can expect to witness an increased usage of UAVs in various long-range applications (e.g., cargo delivery, search-and-rescue, etc.) \cite{b6}.

     Nevertheless, it is challenging to maintain effective communications for cellular-connected UAVs in long-range applications, as cellular networks are not optimized for aerial coverage. In particular, \textit{aerial coverage holes} are scattered throughout the sky, where a UAV's cellular connectivity can be disrupted due to the weak received signal strength from ground base stations (GBSs) \cite{b4}. To address this challenge, exploiting the controllable mobility of a UAV via careful path planning is an attractive solution \cite{b3}. In particular, UAV path planning for long-range applications should minimize the UAV's flying distance to guarantee timely arrival at its designated location. Therefore, it is reasonable for a UAV to fly over some aerial coverage holes in pursuit of a shorter path. Meanwhile, specific communication requirements must be satisfied during a UAV's mission flight, in order to prevent it from losing contact with GBSs due to frequent or long-lasting connectivity outages. A cellular-connected UAV should, therefore, carefully plan its flying paths for long-range missions in a ``connectivity-aware'' manner; i.e., we need to minimize the path length while maintaining effective and reliable cellular connectivity.

    Recently, there have been many studies on UAV-assisted communications, where UAVs serve as base stations or relays and their trajectories are optimized to enhance communication services for terrestrial users \cite{b9}. For cellular-connected UAVs, path planning should, instead, focus on UAV's own mission-specific performance and communication quality along its path. Unfortunately, this problem has not been well studied. The method proposed in \cite{b12} jointly optimizes a UAV's path length, communication latency and interference, but the authors did not consider the cellular connectivity constraints. Continuous connectivity of the UAV with one of the GBSs was enforced in \cite{b13}. However, this is impractical and also unnecessary for long-range UAV applications, given the scattered aerial coverage holes. Although the authors of \cite{b14} considered allowing UAV's temporary disconnection from cellular networks, their algorithm cannot guarantee effective communications during the UAV's flight. Moreover, the cellular coverage models in previous works are oversimplified by assuming line-of-sight (LoS) propagation from GBSs to UAVs, which is not always available in practice.

     In this paper, we investigate connectivity-aware path planning for cellular-connected UAVs in long-range applications. We introduce two new metrics to quantify the cellular connectivity quality of a UAV path. Given these two metrics, communication requirements can be explicitly enforced on UAV path planning, and by adjusting the constraints we can achieve different tradeoffs between the path length and connectivity quality. Additionally, aerial coverage maps are used to provide coverage hole locations, which exploit fine-grained building geometry in modeling ground-to-air propagation. We formulate the path planning problem as finding the shortest path given connectivity constraints, which, however, is NP-hard. To deal with this difficulty, we propose a novel connectivity-aware path planning algorithm with low-complexity based on graph search methods. The aerial coverage map of an urban area in Virginia is built via ray tracing simulations to evaluate our proposed algorithm. Evaluation results show that the proposed algorithm achieves significant performance gains compared with baseline methods, and illustrate the tradeoff between the UAV's path length and connectivity quality.

\section{System Model}\label{s2}

    Fig. \ref{f1} illustrates the cellular-connected UAV communication system investigated in this work. We assume that several GBSs jointly provide cellular coverage for UAV users in a certain region, with potential aerial coverage holes. Moreover, an aerial coverage map is assumed to be available to inform UAVs of coverage hole locations. Under this setting, we consider a UAV is flying at a constant altitude of $H$ meters above the ground, and it is appointed to fly from a source location to a destination location.

    \vspace{-0.1cm}

    \subsection{UAV Motion Model}

        For ease of exposition, we make the following mild assumptions to characterize UAV's dynamics,

      \vspace{-0.2cm}

        \begin{asu}\label{a1}
            (State space discretization): \upshape{The horizontal mobility space of the UAV is discretized into $N$ and $M$ intervals along the $x$ and $y$ axes, respectively, yielding $N \times M$ equally-sized rectangular grids. Each grid unit $\mathbf{u}$ has a set of integer coordinates of the form $(i,j)$, and the set for all grid unit coordinates is denoted as $\mathcal{G} = \{(i,j) | 0 \leqslant i \leqslant N-1, 0 \leqslant j \leqslant M-1, i, j \in \mathbb{Z}\}$. Similar assumptions were made in \cite{b12}.}
        \end{asu}
      \vspace{-0.3cm}
        \begin{asu}\label{a2}
            (Action space discretization): \upshape{The UAV takes discrete steps in one of eight directions (forward, back, left, right and the four diagonal directions) when moving from one grid unit to the neighboring grid unit. The set for all steps is denoted as \small $\mathcal{A} = \{ (0,1), (0,-1), (1,0), (-1,0), (1,1), (1,-1), (-1,1), (-1,-1)\}$.}
        \end{asu}

       Let $\mathbf{u}_k$ denote the coordinates of the grid unit at which the UAV stays at the $k^{th}$ state during its movement. Accordingly, the UAV dynamic model can be formulated as the following state transition equation:

       \vspace{-0.4cm}

        \begin{equation}
          \mathbf{u}_{k+1} = \mathbf{u}_k +a, \quad \mathbf{u}_k, \mathbf{u}_{k+1} \in \mathcal{G}, a \in \mathcal{A}. \label{e1}
            \vspace{-0.1cm}
        \end{equation}
        Thus, the UAV flying path $\mathbf{p}$ is determined by a sequence of $K+1$ two dimensional state-tuples,

        \vspace{-0.6cm}

        \begin{gather}\label{e2}
          \textbf{p} = < \mathbf{u}_k | \mathbf{u}_k \in \mathcal{G}, k = 0,1,2 \ldots K >,
        \end{gather}
        where $K$ is the total number of steps that the UAV needs to take to reach $\mathbf{u}_K$ from $\mathbf{u}_0$. Moreover, the UAV trajectory needs to satisfy the following constraints:

        \vspace{-0.6cm}
        \begin{gather}
          \mathbf{u}_0 = (i_s, j_s), \mathbf{u}_K = (i_d, j_d), \label{e3}\\
          \frac{\|\mathbf{u}_k - \mathbf{u}_{k-1}\|}{\Delta t_{k,k-1}} \leqslant V_{max}, k = 1,2 \ldots K, \label{e5}
            \vspace{-0.6cm}
        \end{gather}
        where $\Delta t_{k,k-1}, V_{max}, (i_s, j_s)$ and $(i_d, j_d)$ denote, respectively, the time duration of the UAV transition from the $k^{th}$ state to the $(k+1)^{th}$ state, maximum UAV velocity and grid coordinates of the initial and final locations, with $\| \cdot \|$ representing the Euclidean distance. While constraint \eqref{e3} restricts the UAV to flying between a given source-destination pair, constraint \eqref{e5} enforces a maximum speed requirement. As \eqref{e5} introduces a rather complex space-time constraint, we further make the following constant speed assumption for tractability.

        \vspace{-0.2cm}

        \begin{asu}\label{a3}
            (Constant speed): \upshape{The UAV is assumed to fly at a constant speed of $V_{const} (V_{const} \leqslant V_{max}) \text{ m/s}$ (as done in \cite{b12,b13}).}
        \end{asu}

        \begin{figure}[!t]

        \centering
        \includegraphics[height=2.2in]{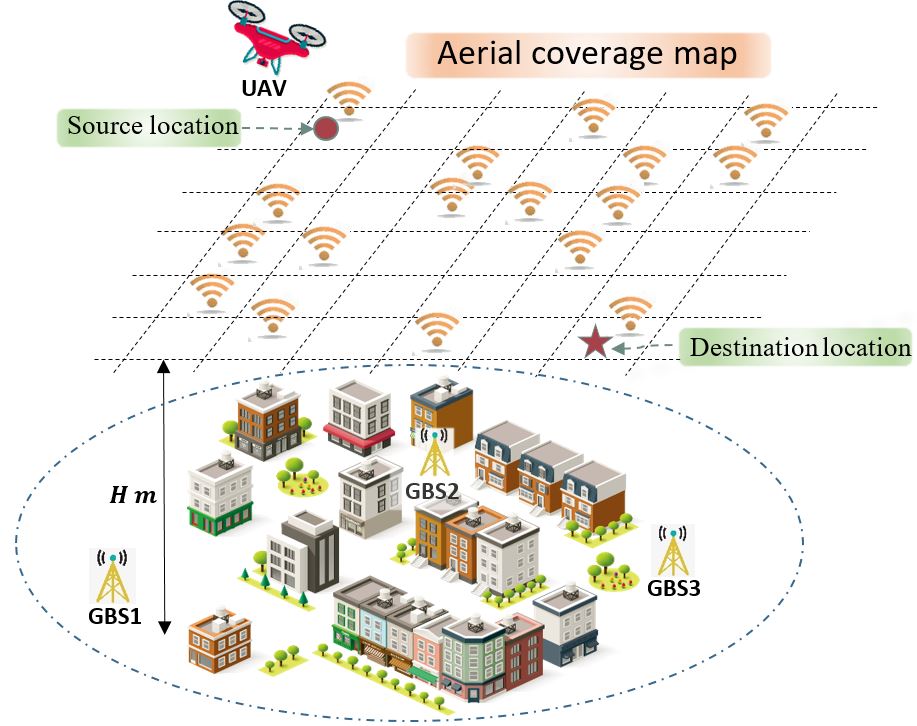}\\
        \caption{A cellular-connected UAV communication system.}
        \vspace{-0.3cm}
        \label{f1}
    \end{figure}

    \vspace{-0.3cm}

    \subsection {Aerial Coverage Map Model }
        \looseness=-1
        To assist path planning, we resort to aerial coverage maps, which can help UAVs to avoid ``holes-in-the-sky'' \cite{b7}. The superiority of the map-based approach has been verified in the positioning of a UAV relay  \cite{b15} and trajectory optimization for UAV base stations  \cite{b16}. In our study, a height- and scenario- dependent 2-D aerial coverage map is used, which characterizes whether each location in the UAV flight plane is covered by cellular networks. In particular, we use Wireless InSite$^\circledR$, a ray tracing simulator, to simulate a received signal power map in an urban environment and then generate an aerial coverage map. Detailed illustrations will be provided in Section \ref{s5}.

        We represent an aerial coverage map as a binary matrix $\mathbf{M} \in \{0,1 \} ^{N\times M}$. In particular, we define $m_{ij} = 1$ to indicate that grid unit $\mathbf{u} = (i,j) \in \mathcal{G}$ is under cellular network coverage, and otherwise $m_{ij} = 0$. Furthermore, for a given UAV path $\mathbf{p}$, we define a $(K+1)- \text{dimensional}$ binary sequence to indicate whether the UAV is under cellular coverage at each of its $K+1$ states, i.e.,
        \vspace{-0.2cm}
        \begin{gather} \small
          \mathbf{c}^\mathbf{p} = <c^\mathbf{p}_k = m_{ij} | (i,j) = \mathbf{u}_k , \mathbf{u}_k \in \mathbf{p}, k = 0, 1, \ldots K>.
        \end{gather}

\section{Connectivity Quality Metrics and Problem Formulation}\label{s3}
        In this section, two new metrics are firstly introduced to quantify the cellular connectivity quality of a UAV path. The connectivity-aware path planning is then formulated as a problem of finding the shortest path subject to constraints on the cellular connectivity quality.
        \begin{figure}[!t]
            \centering
            \includegraphics[height=1.9in]{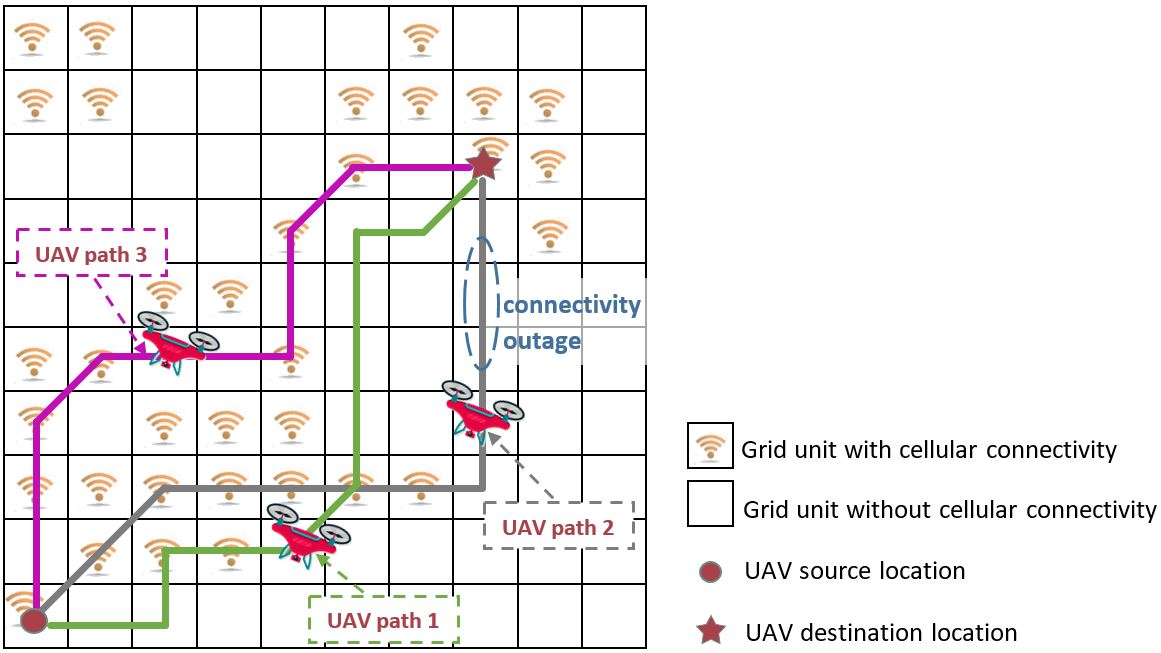}\\
            \caption{An illustration of UAV paths of the same length but with different cellular connectivity quality.}
            \label{f3}
            \vspace{-0.3cm}
        \end{figure}
    \vspace{-0.2cm}
    \subsection{Metrics for Cellular Connectivity Quality} \label{s3_s}

        \looseness=-1
        Consider the example in Fig. \ref{f3}. Obviously, UAV path 1 experiences worse cellular connectivity quality than path 2 and path 3, due to the higher frequency of connectivity outages. In this regard, we introduce the first performance metric to evaluate how often a connectivity outage happens in a UAV path.

        \begin{defn}
            \upshape{The connectivity outage ratio (COR) is defined as the percentage of the grid units that exhibit aerial coverage holes along a UAV path. Specifically, given a coverage indicator sequence $\mathbf{c}^\mathbf{p}$ associated with the UAV path $\mathbf{p}$, the COR is calculated by
            \vspace{-0.3cm}
            \begin{equation}\label{e7}
              COR _\mathbf{p} [\%] = 1 - \frac{\sum _{k=0}^K \mathbf{c}_k^\mathbf{p}}{K+1}  .
            \end{equation}
            }
        \end{defn}

        \begin{rem}\label{r1}
            \looseness=-1
            \upshape{The COR captures the availability of cellular communications along the UAV path. In other words, if a UAV initializes a communication request randomly during its mission flight, the COR represents the probability that the UAV's cellular connectivity will be disrupted and unable to serve this request. }
                \vspace{-0.2cm}
        \end{rem}

        \looseness=-1
        Next, we note that although with the same COR value, intuitively, the connectivity quality of UAV path 3 in Fig. \ref{f3} is better than that of path 2, as the UAV in path 2 needs to stay longer in one uncovered area. Based on this observation, we define the second performance metric, a measurement of how long each connectivity outage lasts in a UAV path.

        \begin{defn}
            \upshape{The connectivity outage duration (COD) refers to the length of the consecutive coverage holes on the UAV's path.\footnote{Note that a UAV path may include several COD values, as a UAV may fly through aerial coverage holes intermittently.} Given a coverage indicator sequence $\mathbf{c}^\mathbf{p}$ associated with the UAV path $\mathbf{p}$, we define a set of CODs on this path as $\{ COD_\mathbf{p} ^{(\ell)} : \ell = 1,2,\ldots, L\}$, where $L$ represents the total number of times that a connectivity outage happens. $COD_\mathbf{p} ^{(\ell)}$ is calculated by
            \vspace{-0.5cm}
              \begin{gather}\small \label{e8}
                COD_\mathbf{p} ^{(\ell)} = \sum _{k = i^ {(\ell)}}^ {j^ {(\ell)}} \| \mathbf{u}_k - \mathbf{u}_{k-1} \|, \notag\\
                \text{ if } \sum _{k = i^ {(\ell)}}^ {j^ {(\ell)}} c_k^\mathbf{p} = 0, 1\leqslant i^ {(\ell)} \leqslant j^ {(\ell)} \leqslant K  .
              \end{gather}
             }
        \end{defn}

        \begin{rem}\label{r2}
            \upshape{The COD is a metric related to the communication latency in cellular-connected UAV communications. Specifically, if a UAV attempts to send data to its GCS while its cellular connectivity is in an outage, data transmission has to be delayed until the UAV reconnects to cellular networks. In this case, the delay is upper bounded by the duration of the connectivity outage. }
        \end{rem}
    \vspace{-0.3cm}
    \subsection{Connectivity-Aware Path Planning Problem}

    In this paper, we investigate cellular-connected UAVs for long-range applications, and consider a connectivity-aware path planning problem. We aim to minimize the path length to save energy and reduce the mission completion time, while maintaining effective UAV communications during the UAV's mission flight. The COD and COR metrics are adopted to enforce communication constraints for UAV path planning. Consequently, the connectivity-aware path planning is formulated as a problem of finding the shortest path under constraints on the value of COR and COD:
    \begin{align}
      \mathcal{P}_1: \quad \min _\mathbf{p} \quad & \sum_{k=1}^{K}\| \mathbf{u}_{k} - \mathbf{u}_{k-1} \|, \label{e10}  \\
      s.t. \quad & \eqref{e1}, \eqref{e3}, \notag \\
      & COD_\mathbf{p} ^{(\ell)} \leqslant d, \quad \forall \ell = 1,2,\ldots, L, \label{e13}\\
      & COR_\mathbf{p} \leqslant r \label{e14},
      \vspace{-0.2cm}
    \end{align}
    where $d$ $(d\geqslant 0)$ and $r$ $(0\leqslant r < 1)$ denote, respectively, the maximum tolerant COD and COR of a UAV path. Thus, constraints \eqref{e13} and \eqref{e14} ensure that the cellular connectivity quality of a UAV's path satisfies the designated requirements.\footnote{\looseness=-1 The connectivity constraints can be adjusted according to different targeted communication performance in different UAV applications. Some quantitative communication requirements in civil UAV applications are provided in \cite{b6}.}

    Problem $\mathcal{P}_1$ is a constrained shortest path problem, which is generally NP-hard. In particular, the main difficulty is introduced by the two connectivity constraints. In the next section, we will firstly propose two methods to handle the COD and COR constraints, respectively, and then develop a low-complexity algorithm to find a heuristic solution to $\mathcal{P}_1$.

\section{Proposed Path Planning Algorithms}\label{s4}

    In this section, we present a heuristic algorithm for problem $\mathcal{P}_1$ based on graph search methods. We start by defining an undirected graph $G = (V, E)$, where the node set $V$ represents all the grid units in the UAV state space, and the edge set $E$ represents their connections. Hence, we have
    \begin{align}
      V &= \mathcal{G}, \label{e15}\\
      E &= \{e=(\mathbf{u},\mathbf{v}) | \mathbf{u},\mathbf{v} \in V, \mathbf{v} = \mathbf{u} + a, a \in A\}. \label{e15_1}
    \end{align}
    We also associate each edge $e \in E$ with a weight $l(e)$. In order to integrate the cost of the UAV flying distance into the graph representation, we define $l(e)$ as
    \begin{equation}\label{e16}
      l(e) =
      \begin{cases}
	       d_1 & \text{if } \mathbf{v} = \mathbf{u} + x, x \in \mathcal{D}\\
	       d_2 & \text{if } \mathbf{v} = \mathbf{u} + x, x \in \mathcal{H},
	   \end{cases}
    \end{equation}
    where $d_1$ is the distance cost of moving vertically or horizontally on the 2-D grid, and $d_2$ is the distance of diagonal movement, with $\mathcal{D} = \{(0,1), (0, -1), (1, 0), (-1, 0)\}, \mathcal{H} = \{(1, 1), (1, -1), (-1, 1), (-1, -1)\}$ denoting two subspaces of action space $\mathcal{A}$. Next, prior to solving problem $\mathcal{P}_1$, we propose two algorithms to solve, respectively, the UAV's trajectory optimization problems subject to only the COD constraint (denoted as $\mathcal{P}_2$) and only the COR constraint (denoted as $\mathcal{P}_3$).

    \subsection{Modified A* Algorithm for $\mathcal{P}_2$}
    \looseness=-1
    Based on the classical A* shortest path algorithm \cite{b20}, we propose a modified A* search algorithm to handle the COD constraint, and thus solve problem $\mathcal{P}_2$. Specifically, in our proposed algorithm, each node $\mathbf{v} \in V$ is evaluated by the following value:
    \vspace{-0.3cm}
    \begin{equation}\label{e19}
      w(\mathbf{v}) = g(\mathbf{v}) + h(\mathbf{v}),
      \vspace{-0.2cm}
    \end{equation}
    where $g(\mathbf{v})$ is the true optimal cost from the given source node $ \mathbf{s} = (i_s, j_s)$ to the current node $\mathbf{v}$, and $h(\mathbf{v})$ is a heuristic estimate of the cost from $\mathbf{v}$  to the destination node $ \mathbf{d} = (i_d, j_d)$. In particular, the true optimal cost from an initial node to the current node $\mathbf{v}$ can be obtained by summing the edge weights over the shortest path to reach $\mathbf{v}$. Consequently, the calculation of $g(\mathbf{v})$ is given by the following iterative form:
    \vspace{-0.2cm}
        \begin{align}
          &g(\mathbf{s}) = 0 \label{e20},\\
          &g(\mathbf{v}) = g(\mathbf{u}) + l((\mathbf{u},\mathbf{v})), \label{e21}
        \end{align}
    where $\mathbf{u}$ is the last node in the shortest path from $\mathbf{s}$ to $\mathbf{v}$.
        Moreover, we adopt the octile distance as an estimate of the optimal cost from $\mathbf{v}$ to $ \mathbf{d}$, that is,
        \begin{align}\label{e22}
          h((i,j)) = & d_1 \times (|i_d - i|+|j_d - j|) + (d_2 \notag \\
           &  - 2 \times d_1) \times min((|i_d - i|,|j_d - j|)), (i,j) \in V.
        \end{align}

       \begin{algorithm}[!t]\small
		  \caption{Modified A* search algorithm for $\mathcal{P}_2$}
		  \begin{algorithmic}[1]
		      \renewcommand{\algorithmicrequire}{\textbf{Input:}}
		      \REQUIRE UAV state space $\mathcal{G}$, aerial coverage matrix $\mathbf{M}$, source coordinates $\mathbf{s}$, destination coordinates $ \mathbf{d} $ and COD threshold $d$.
              \STATE \textbf{Initialize}:  create a graph using \eqref{e15}--\eqref{e16} and a min priority queue $Q$ with $w$ as key, $\mathbf{s}.parent \gets nil$, $d_{zero} (\mathbf{s})\gets 0$, \textbf{Decrease-Key} $ (Q, \mathbf{s}, w(\mathbf{s}), d_{zero} (\mathbf{s})) $;
			  \WHILE{$Q \neq \emptyset$ }
			     \STATE $ \mathbf{u}\gets \textbf{Extract-Min} (Q)$;
                 \IF {$\mathbf{u} = \mathbf{t} $ }
			         \STATE $\mathbf{p} \gets \textbf{BackTrace}(\mathbf{u})$;
			         \RETURN SUCCESS, Path $\mathbf{p}$
                 \ENDIF
			  \STATE Mark $\mathbf{u}$ as visited.
              \FORALL{$a \in \mathcal{A}$}
			     \STATE $\mathbf{v} \gets \mathbf{u} + a$;
			         \IF {$\mathbf{v}$ not visited \AND $d_{zero}(\mathbf{v})\leqslant d $}
			             \STATE $\mathbf{v}.parent \gets \mathbf{u}$, \textbf{Decrease-Key} $ (Q, \mathbf{s}, w(\mathbf{s}), d_{zero} (\mathbf{s}))$;
                      \ENDIF
			  \ENDFOR
			  \ENDWHILE
			  \RETURN FAILURE
		  \end{algorithmic}

	    \end{algorithm}
\vspace{-0.3cm}

    \begin{algorithm}[!t]\small
		  \caption{Weight variation algorithm for $\mathcal{P}_3$}
		  \begin{algorithmic}[1]
		      \renewcommand{\algorithmicrequire}{\textbf{Input:}}
		      \REQUIRE UAV state space $\mathcal{G}$, aerial coverage matrix $\mathbf{M}$, source coordinates $\mathbf{s} $, destination coordinates $ \mathbf{d}$ and COR threshold $r$.
              \STATE \textbf{Initialize}: $\delta = \bar{\delta} + 1$, create a graph using \eqref{e15}, \eqref{e15_1}, \eqref{e27}.
              \STATE  Find the shortest path $\mathbf{p}$ from $\mathbf{s}$ to $\mathbf{d}$ on graph $G$ via A* search algorithm.
              \IF {$COR_\mathbf{p} > r$ }
                \RETURN FAILURE
              \ELSE
                \REPEAT
                    \STATE $COR_{old} \gets COR_\mathbf{p}$;
                    \STATE $\delta \gets \delta / 2$;
                    \STATE Update edge weights of $G$ using equation \eqref{e27};
                    \STATE Find the shortest path $\mathbf{p}$ from $\mathbf{s}$ to $\mathbf{d}$ on graph $G$ via A* search algorithm.
                \UNTIL {$COR_\mathbf{p} >r$}
                \RETURN SUCCESS, Path $\mathbf{p}$
              \ENDIF
		  \end{algorithmic}
	    \end{algorithm}

   Under this setting, the proposed algorithm proceeds as follows. Similar to the classical A* search, it also maintains a priority queue $Q$ where a candidate node with a lower evaluation value $w$ is given a higher priority. Traditionally, at each iteration, the $\mathbf{Extract-Min} (Q)$ method pops out the node $\mathbf{v}_{pop}$ at the front of $Q$, and then unvisited neighboring nodes of $\mathbf{v}_{pop}$ will be added into $Q$. The key to the proposed algorithm is to perform a COD constraint feasibility check before inserting a node into $Q$, so as to ensure the COD of the updated path after the current iteration does not exceed $d$. To this end, we introduce a variable $d_{zero}$ for each unvisited neighboring node $(i,j)$ of $\mathbf{v}_{pop}$. The value of $d_{zero} ((i,j))$ is assigned to be the COD introduced by adding node $(i,j)$ into the updated path, which is
   \vspace{-0.2cm}
        \begin{equation}\label{e23} \small
          d_{zero} ((i,j)) =
          \begin{cases}
	       0 & \text{if } m_{ij} = 1, \\
	       d_{zero} (\mathbf{v}_{pop}) + l((i,j), \mathbf{v}_{pop}) & \text{if } m_{ij} = 0.
	      \end{cases}
        \end{equation}
    \vspace{-0.2cm}
This algorithm is summarized in \textbf{Algorithm 1}.

\begin{figure*}[!t]
        \centering
        \subfigure[]{
        \begin{minipage}[b]{0.3\textwidth}
            \centering
            \includegraphics[width=5cm]{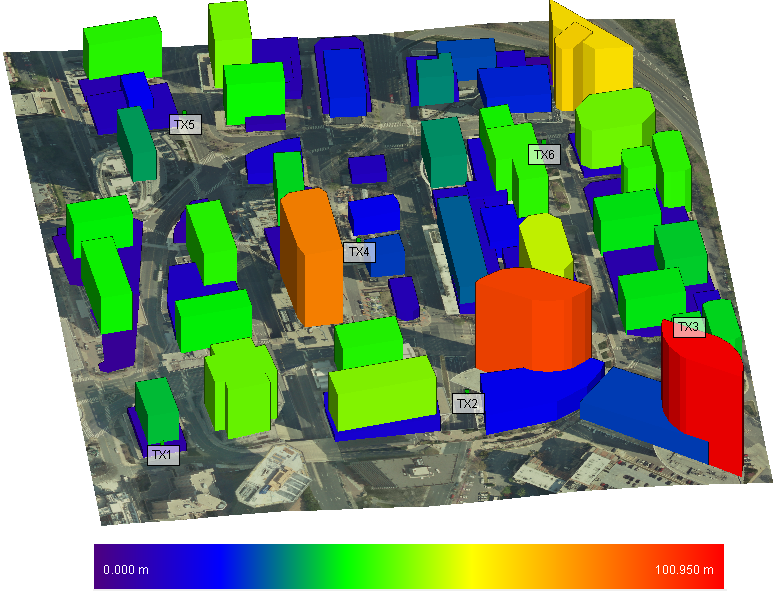}
        \end{minipage}
        }
        \subfigure[]{
        \begin{minipage}[b]{0.3\textwidth}
            \centering
            \includegraphics[width=4.1cm]{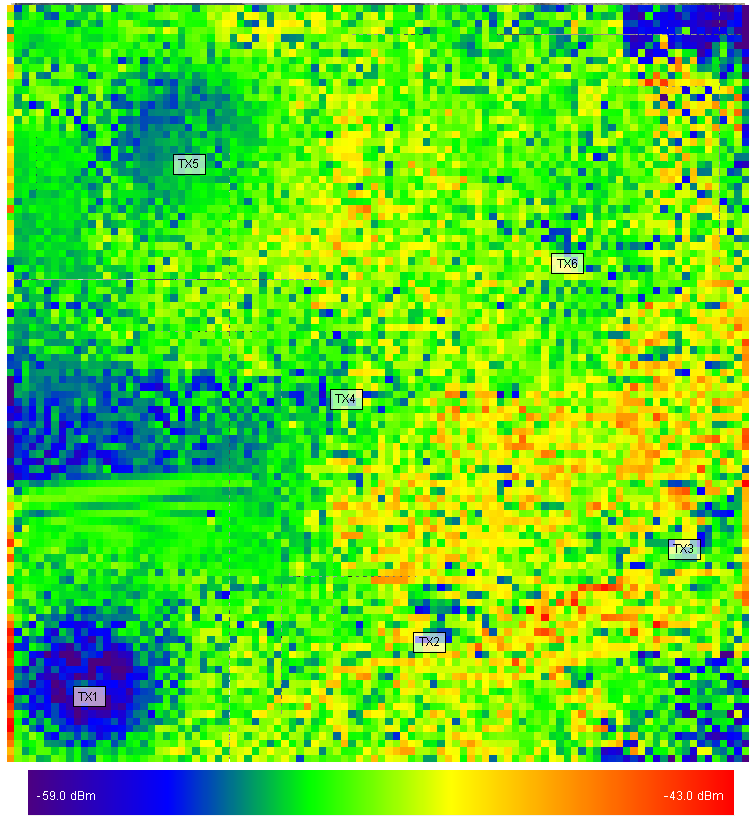}
            \end{minipage}
        }
        \subfigure[]{
        \begin{minipage}[b]{0.3\textwidth}
            \centering
            \includegraphics[width=4.3cm]{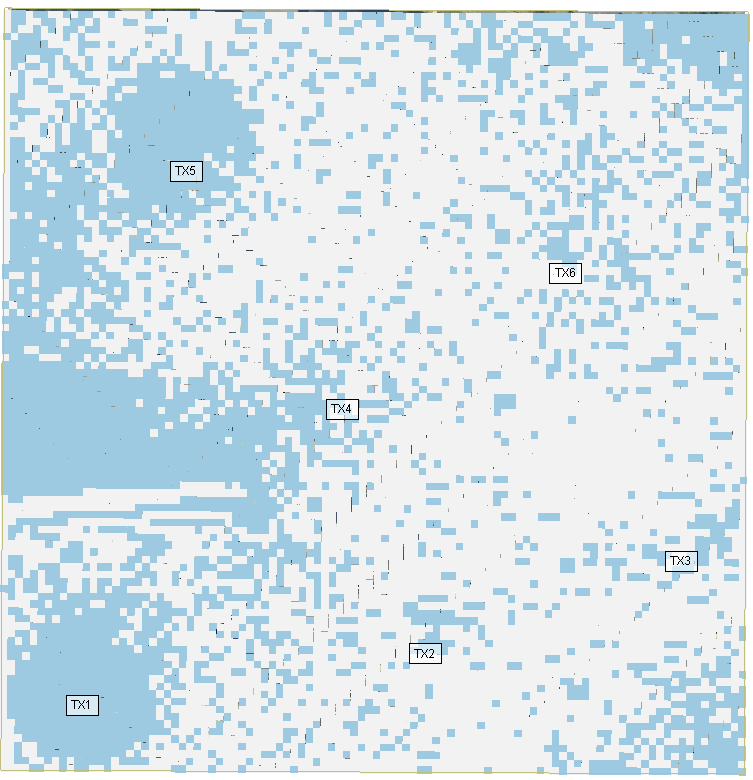}
            \end{minipage}
        }
        \caption{\looseness=-1 (a) An environment map of a section od Virginia, with different colors representing different building heights. (b) An RSP map corresponding to $100 \times 102$ discretized locations on the UAV flight plane with an altitude of $110 \ m$. (c) A cellular coverage map of the region, with blue denoting coverage holes.}
        \vspace{-0.3cm}
        \label{f4}
\end{figure*}

    \subsection{Weight Variation Algorithm for $\mathcal{P}_3$}

        The COR constraint is a global one, which makes problem $\mathcal{P}_3$ difficult to solve optimally.  As the algorithm in the previous subsection cannot be applied, we present another algorithm to handle the COR constraint and solve $\mathcal{P}_3$ efficiently.

        Clearly, including more aerial coverage holes in the path results in a higher COR value. Thus, we propose to impose a penalty to prevent the UAV from frequently visiting coverage holes. To this end, we increase the weight of edges connecting to nodes without cellular connectivity by a certain value $\delta$ $(\delta\geqslant 0)$,  i.e.,
        \vspace{-0.2cm}
        \begin{equation} \small \label{e27}
          l(e) =
            \begin{cases}
	           d_1 & \text{if } \mathbf{v} = \mathbf{u} + x,  m_{\mathbf{v}.i,\mathbf{v}.j} = 1 \text{ and }  x \in \mathcal{D},\\
               d_1 + \delta & \text{if } \mathbf{v} = \mathbf{u} + x,  m_{\mathbf{v}.i,\mathbf{v}.j} = 0 \text{ and } x \in \mathcal{D},\\
	           d_2 & \text{if } \mathbf{v} = \mathbf{u} + x , m_{\mathbf{v}.i,\mathbf{v}.j} = 1 \text{ and }  x \in \mathcal{H},\\
               d_2 + \delta & \text{if } \mathbf{v} = \mathbf{u} + x , m_{\mathbf{v}.i,\mathbf{v}.j} = 0 \text{ and } x \in \mathcal{H}.\\
	        \end{cases}
        \end{equation}

        \looseness=-1
        Additionally, in order to reduce the search complexity, we impose another requirement on the UAV path, namely, \emph{each state in the path sequence $\mathbf{p}$ must be different to others}. This implies that cycles are not allowed within the path, and thus largely reduces the search space on the graph. Meanwhile, it prevents UAVs from traversing to and from grid units with cellular connectivity to decrease the COR of its path, which generally makes no sense in practice. Therefore, although enforcing this additional requirement leads to no guarantee of finding a feasible path to satisfy the COR constraint, it is still reasonable for UAV path planning. With this requirement, the solution of problem $\mathcal{P}_3$ can be obtained by: i) assigning different values of $\delta$ to vary the edge weights on the graph; ii) for each weight variation, applying existing shortest path algorithms (e.g., the Dijkstra's algorithm and A* algorithm \cite{b20}) to search a path with the lowest cumulative weights; and iii) selecting the path that has the minimal length as well as satisfying the COR constraint among all the obtained paths. In this way, $\mathcal{P}_3$ is equivalent to finding an optimal value $\delta^ \star$ that leads to the shortest path algorithm returning the desired path. Next, we provide two lemmas which are useful to efficiently determine $\delta^ \star$.
        \begin{lem}
          \upshape{When $\delta > \bar{\delta} = N \times M$, the shortest path algorithm will find a UAV path with the minimum COR among all paths connecting the source-destination pair.}
        \end{lem}

        \begin{proof}
          \upshape{The proof is omitted due to space limitation.}
        \vspace{-0.1cm}
        \end{proof}

        \begin{lem}
          \vspace{-0.1cm}
          \upshape{Suppose that for $\delta_1$ and $\delta_2$ (with $\delta_1 > \delta_2$), the shortest path algorithm returns, respectively, path $\mathbf{p}_1$ and $\mathbf{p}_2$. We then have $COR_{\mathbf{p}_1} \leqslant COR_{\mathbf{p}_2}$, and the length of $\mathbf{p}_2$ is no longer than $\mathbf{p}_1$.}
        \end{lem}
        \begin{proof}
        \vspace{-0.1cm}
          \upshape{The proof is omitted due to space limitation.}
          \vspace{-0.1cm}
        \end{proof}

        \vspace{-0.1cm}
        \looseness=-1
        Based on Lemmas 1 and 2, we are ready to propose the algorithm for $\mathcal{P}_3$, which proceeds as follows. Initially, we check the feasibility of $\mathcal{P}_3$ by setting $\delta = \bar{\delta} + 1$ and obtain a lower bound for the achievable COR value. If this lower bound does not satisfy the COR constraint, the algorithm fails to return a feasible solution to $\mathcal{P}_3$. Otherwise, we use a binary search to find the optimal $\delta ^\star$ efficiently. Then, we apply the A* shortest path algorithm on graph $G$ where some edge weights are increased by $\delta ^\star$, and the resulting path gives a heuristic solution to $\mathcal{P}_3$. The proposed algorithm is summarized in \textbf{Algorithm 2}.

    \subsection{Connectivity-Aware Path Planning Algorithm }
        In order to find a solution to the original problem $\mathcal{P}_1$ with both COD and COR constraints, we propose a new connectivity-aware path planning scheme, which integrates Algorithm 1 into Algorithm 2 as a subroutine. Specifically, we replace the A* search algorithm adopted in Algorithm 2 with our modified version. In this case, at each round with $\delta$ lessened to half in Algorithm 2, the updated path always satisfies the COD constraint, and thus the successfully returned path will have the desired connectivity quality as well as a minimized distance. Next, we analyze the computational complexity of the proposed UAV path planning algorithm.

        The running time of Algorithm 1 is the same as that of the A* search algorithm ($\mathcal{O} (NM \log (NM))$ in our case). The binary search in Algorithm 2 runs in $\mathcal{O} (\log  (NM))$ iterations. Realizing that Algorithm 1 is executed at each iteration in Algorithm 2, the total running time of our proposed path planning algorithm is given by $\mathcal{O} (NM \log ^2 (NM))$.

\section{Simulation Results}\label{s5}

    Fig. \ref{f4} provides an example of constructing an aerial coverage map via ray tracing simulations. The building geometry and GBS locations of the simulation scenario are shown in Fig. \ref{f4} (a), which corresponds to a section in Rosslyn, Virginia, USA, with a size of approximately $494\ m \times 507\ m$. The ray tracing software Wireless InSite$^\circledR$ is used to simulate the fine-grained signal propagation from GBSs to UAVs in this target area. Specifically, we uniformly placed receivers at $10200$ $(100 \times 102)$ locations on the UAV flight plane at an altitude of $110 \ m$, and then simulated the received signal power (RSP) from each GBS. The strongest simulated RSP values are shown in Fig. \ref{f4} (b). By setting an RSP threshold as $-52.6$ dBm to ensure coverage, we binarize the RSP map into a cellular coverage map, as shown in Fig. \ref{f4} (c). The obtained coverage map clearly illustrates the non-uniformity of the aerial coverage, affected by the building blockage. This indicates that careful path planning is critical for maintaining good communication quality, and the effect of environment/buildings should be carefully considered.

    Next, we use this coverage map to evaluate the performance of our proposed connectivity-aware UAV path planning algorithm. Suppose a UAV, located at $(4,17)$, needs to fly to a destination at coordinates $(92,94)$, with the COD and CDR constraints set as $d=3$ and $r=10\%$, respectively. We normalize the distance on the grid map for ease of illustration, i.e., $d_1 = 1$ and $d_ 2 = 1.4$ for \eqref{e16}.

    Fig. \ref{f5} illustrates the paths obtained by three different methods. For comparison, we consider two other UAV path planning schemes, where the naive shortest path algorithm pursues the minimal flying distance without considering connectivity constraints, and the coverage hole detour scheme returns a path without any coverage holes. It can be observed that our proposed connectivity-aware path planning approach  selects a path towards the destination in a smart way. Specifically, it does not deviate significantly from the naive shortest path, so as to shorten the UAV's mission path. Meanwhile, it avoids some unfavorable coverage holes to have the desired cellular connectivity quality. Table \ref{t1} shows the performance of three different paths in Fig. \ref{f5}. We see that our proposed algorithm achieves an up to $83.3\%$ decrease in the maximum COD and a significant reduction of $87.5\%$ in the COR, compared with the naive shortest path. Moreover, such a connectivity performance improvement is at a low expense of the UAV's flying distance. In particular, our proposed algorithm slightly increases the path length by $8.2\%$. In contrast, the coverage hole detour path, which avoids all coverage holes, increases the path length by $22.2\%$. This illustrates the effectiveness of the proposed connectivity-aware path planning.

    We run \textbf{Algorithm 1} with different COD constraints for the source-destination location pair of $(6, 28)$ and $(33, 93)$, and the results of the obtained path length versus the maximum COD on the path are shown by the red dots in Fig. \ref{f6} (a). Similarly, in Fig. \ref{f6} (b), we plot the path length versus the COR that results from setting various targets of the COR metric in \textbf{Algorithm 2}. Generally, the curves in Fig. \ref{f6}, which are fitted with the simulated values, illustrate that there is a tradeoff between the path length and the communication quality. This result is expected as the trajectories of UAVs usually become zigzag-like to avoid the coverage holes when a more stringent cellular connectivity requirement is enforced. Using our proposed connectivity-aware path planning algorithm, designers can adjust the connectivity constraints according to different application scenarios and achieve the desired tradeoff.

\begin{figure}[!t]
            \centering
            \includegraphics[height=2in]{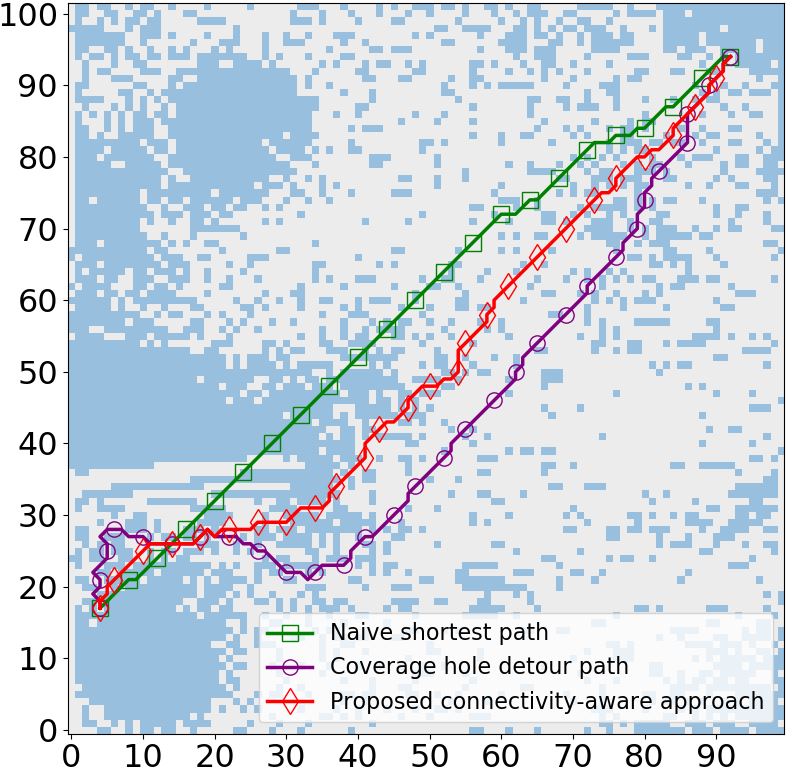}\\
            \caption{Illustration of UAV paths obtained by three different methods, with constraints on COD and COR set as 3 and 10\%, respectively.}
            \vspace{-0.3cm}
            \label{f5}
    \end{figure}

\begin{table}[!t] \scriptsize
    \vspace{-0.1cm}
     \caption{Performance assessment for paths in Fig. 4}
     \begin{center}
     \begin{tabular}{|c|c|c|c|}
          \hline
           & Path length & COR & Maximum COD\\
          \hline \hline
          Naive shortest path & 118.8 & $46.07\%$ & 16.8 \\
          Coverage hole detour path & 145.2 & 0 & 0 \\
          Proposed connectivity-aware path & 128.6 & $5.77\%$ & 2.8 \\
          \hline
     \end{tabular}
     \label{t1}
     \end{center}
     \vspace{-0.3cm}
     \end{table}

    \begin{figure}[!t]
        \centering
        \subfigure[]{
        \begin{minipage}[b]{0.225\textwidth}
            \centering
            \includegraphics[height=3.5cm]{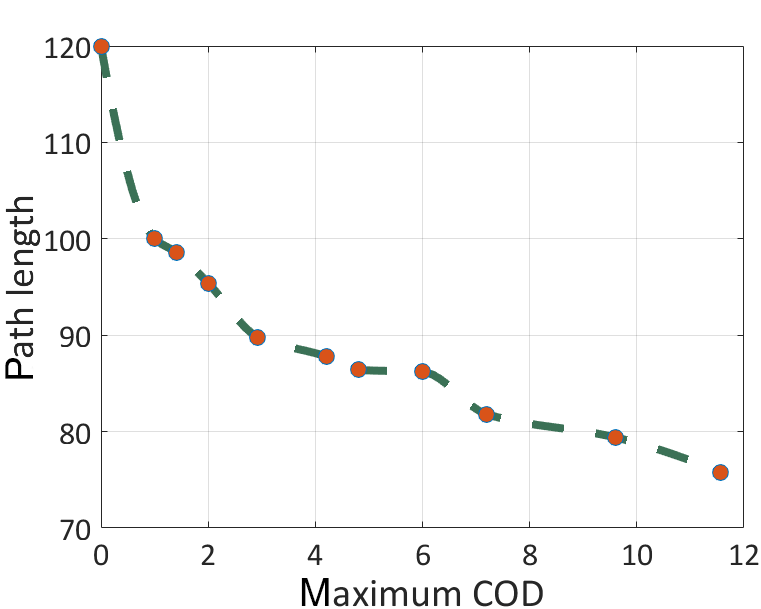}
        \end{minipage}
        }
        \subfigure[]{
        \begin{minipage}[b]{0.225\textwidth}
            \centering
            \includegraphics[height=3.5cm]{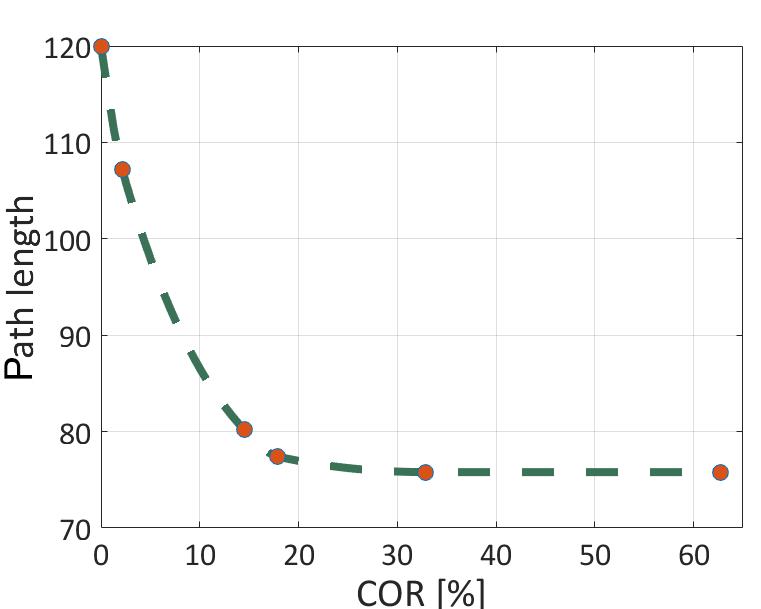}
            \end{minipage}
        }
        \caption{The tradeoff between the path length and communication quality. }
        \vspace{-0.3cm}
        \label{f6}
    \end{figure}

\section{Conclusions}\label{s6}

     In this paper, we investigated a connectivity-aware path planning problem for cellular-connected UAVs. Two new metrics were introduced to quantify the cellular connectivity quality of UAVs, which help to explicitly enforce connectivity constraints on the shortest path finding problem. As the formulated problem is NP-hard, we proposed a low-complexity path planning algorithm based on graph search methods. This study demonstrated the effectiveness of exploiting the controllable mobility of UAVs to satisfy the communication requirements, as well as the importance of considering realistic aerial coverage. For future investigations, it would be interesting to extend this work to online UAV path planning and consider dynamic aerial coverage.

\bibliographystyle{ieeetran}
\bibliography{reference}

\end{document}